\documentclass[a4paper,12pt]{article}
\pdfoutput=1
\usepackage{a4wide}
\usepackage{amssymb, amsmath}
\usepackage[amsmath,hyperref,thmmarks]{ntheorem}
\usepackage{t1enc,times}
\usepackage[utf8]{inputenc}
\usepackage[english]{babel}
\usepackage{enumerate}
\usepackage{hyperref}
\usepackage{natbib}
\usepackage{color,caption}
\usepackage[linesnumbered,lined,boxed,algosection,algo2e]{algorithm2e}
\DontPrintSemicolon

\emergencystretch=\hsize
\def\cc{{\mathcal C}}

\def\cf{{\mathcal F}}
\def\cg{{\mathcal G}}
\def\ch{{\mathcal H}}

\def\cl{{\mathcal L}}

\def\ct{{\mathcal T}}

\def\E{{\mathbb E}}

\def\L{{\mathbb L}}
\def\N{{\mathbb N}}
\def\P{{\mathbb P}}

\def\R{{\mathbb R}}

\def\ee{{\bf e}}
\def\hl{\widehat{\lambda}}
\def\hL{\widehat{\Lambda}}
\def\L{\Lambda}
\def\l{\lambda}
\def\htau{\widehat{\tau}}

\def\ind#1{{\bf 1}_{\left\{#1\right\}}}

\def\abs#1{\left|#1\right|}

\def\Var{\mathop{\rm Var}\nolimits}

\def\norm#1{\mathop{\left\| #1 \right\|}\nolimits}
\def\inv#1{\mathop{\frac{1}{ #1}}\nolimits}
\def\expp#1{\mathop {\mathrm{e}^{ #1}}}

\DeclareMathOperator{\Span}{span}
\DeclareMathOperator{\ClosedSpan}{\overline{span}^{L^2(\Omega, \cf_T)}}
\newcommand{\JL}[1]{\textcolor{blue}{#1}}

\newcommand{\spaceIndex}[1][]{
  \ifthenelse{\equal{#1}{}}
    {A^{\otimes d}_{p,n}}
    {A^{\otimes d}_{p,n|#1}}
}
\theoremstyle{plain}
\newtheorem{theorem}{Theorem}[section]
\newtheorem{proposition}[theorem]{Proposition}
\newtheorem{lemma}[theorem]{Lemma}

\newtheorem{remark}[theorem]{Remark}

\theoremstyle{nonumberplain}
\theoremheaderfont{\normalfont\itshape}
\theorembodyfont{\normalfont}
\theoremseparator{.}
\theoremsymbol{\ensuremath{\blacksquare}}
\newtheorem{proof}{Proof}

\title{Pricing path-dependent Bermudan options using Wiener chaos expansion: an embarrassingly parallel approach\thanks{The High Performance Computations presented in this paper were performed using the Froggy platform of the CIMENT infrastructure (https://ciment.ujf-grenoble.fr), which is supported by the Rhône-Alpes region (GRANT CPER07\_13 CIRA) and the Equip@Meso project (reference ANR-10-EQPX-29-01) of the programme Investissements d'Avenir supervised by the Agence Nationale pour la Recherche.}}
\date{\today}
\author{Jérôme Lelong
  \thanks{Univ. Grenoble Alpes, CNRS, Grenoble INP, LJK, 38000 Grenoble, France. \newline
\indent email: \texttt{jerome.lelong@univ-grenoble-alpes.fr}}}

\begin{document}
\maketitle

\begin{abstract}
	In this work, we propose a new policy iteration algorithm for pricing Bermudan options when the payoff process cannot be written as a function of a lifted Markov process. Our approach is based on a modification of the well-known Longstaff Schwartz algorithm, in which we basically replace the standard least square regression by a Wiener chaos expansion. Not only does it allow us to deal with a non Markovian setting, but it also breaks the bottleneck induced by the least square regression as the coefficients of the chaos expansion are given by scalar products on the $L^2(\Omega)$ space and can therefore be approximated by independent Monte Carlo computations. This key feature enables us to propose an embarrassingly parallel algorithm to efficiently handle non Markovian payoff.


\end{abstract}

\section{Introduction}
\label{sec:intro}

We fix some finite time horizon $T>0$ and a filtered probability space $(\Omega, \cf, (\cf_t)_{0 \le t \le T}, \P)$, where $(\cf_t)_{0 \le t \le T}$ is supposed to be the natural augmented filtration of a $d-$dimensional Brownian motion $B$.  On this space, we consider an adapted process $(S_t)_{0 \le t \le T}$ with values in $\R^{d'}$ modeling a $d'$--dimensional underlying asset, with $d' \le d$. The number of assets $d'$ can be strictly smaller than the dimension $d$ of the Brownian motion to encompass the case of stochastic volatility models or stochastic interest rates. We assume that  $\P$ is a risk neutral measure. We consider a Bermudan option with exercising dates $0 = t_0 \le T_1 < T_2 < \dots < T_N = T$ and discounted payoff $Z_{T_k}$ if exercised at time $T_k$.  We assume that the discrete time payoff process $(Z_{T_k})_{0 \le k \le N}$ is adapted to the filtration $(\cf_{T_k})_{0 \le k \le N}$ and satisfies $\max_{0 \le k \le N } \abs{Z_{T_k}} \in L^2$. This framework naturally encompasses the case of path-dependent options, ie. when the payoff process writes $\tilde Z_{T_k} = \phi_k((S_u; 0 \le u \le T_k))$ for any $0 \le k \le N$.

Standard arbitrage pricing theory defines the discounted value of the Bermudan option at times $(T_k)_{0 \le k \N}$ by
\begin{equation}
  \label{eq:dpp}
  \begin{cases}
    U_{T_N} & = Z_{T_N} \\
    U_{T_k} & = \max\left( Z_{T_k}, \E[U_{T_{k+1}} | \cf_{T_k}] \right)
  \end{cases}
\end{equation}
Solving this backward recursion known as the dynamic programming principle has been a challenging problem for years and various approaches have been proposed to approximate its solution. The real difficulty lies in the computation of the conditional expectation $\E[U_{T_{k+1}} | \cf_{T_k}]$ at each time step of the recursion. If we were to classify the different approaches, we could say that there are regression based approaches (see \cite{carriere96,tsit:vanr:01} and quantization approaches (see \cite{bapa03,pages13_quantization}). We refer to \cite{BoucharWarin12} and \cite{pages18numerical} for a survey of the different techniques to price Bermudan options.

Among all the available algorithms to compute $U$ using the dynamic programming principle, the one proposed by~\cite{LS01} has the favour of practitioners. Their approach is based on iteratively selecting the optimal policy. Let $\tau_k$ be the smallest optimal policy after time $T_k$, then
\begin{equation}
  \label{eq:policy-iteration}
  \begin{cases}
    \tau_N = T_N \\
    \tau_k = T_k \ind{Z_{T_k} \ge \E[Z_{\tau_{k+1}} | \cf_{T_k}]} + \tau_{k+1} \ind{Z_{T_k} < \E[Z_{\tau_{k+1}} | \cf_{T_k}]}, \; \text{for $1 \le k \le N-1$}
  \end{cases}
\end{equation}
All these methods based on the dynamic programming principle either as value iteration~\eqref{eq:dpp} or policy iteration~\eqref{eq:policy-iteration} require a Markovian setting to be implemented such that the conditional expectation knowing the whole past can be replaced by the conditional expectation knowing only the value of a Markov process at the current date.
The theory of the Snell envelope states that the sequence $U$ also satisfies
\begin{equation}
  \label{eq:price}
  U_{T_k} = \sup_{\tau \in \ct_{T_k, T}} \E[Z_{\tau} | \cf_{T_k}].
\end{equation}

When the discounted payoff process writes $Z_{T_k} =  \phi_k(X_{T_k})$, for any $0 \le k \le N$, where $(X_t)_{0 \le t \le T}$ is an adapted Markov process, the conditional expectation involved in~\eqref{eq:policy-iteration} simplifies into
\begin{equation}
  \label{eq:EMarkov}
 \E[Z_{\tau_{k+1}} | \cf_{T_k}]  = \E[Z_{\tau_{k+1}} | X_{T_k}] = \psi_k(X_{T_k})
\end{equation}
where $\psi_k$ solves the following minimization problem
\begin{equation*}
  \inf_{\psi \in L^2(\cl(X_{T_k}))} \E\left[\abs{Z_{\tau_{k+1}} - \psi(X_{T_k})}^2\right]
\end{equation*}
with $L^2(\cl(X_{T_k}))$ being the set of all measurable functions $f$ such that $\E[f(X_{T_k})^2] < \infty$.
The real challenge comes from properly approximating the space $L^2(\cl(X_{T_k}))$ by a finite dimensional vector space: one typically uses polynomials or local bases. In both cases, to ensure a decent accuracy, the dimension of the approximation of $L^2(\cl(X_{T_k}))$ increases exponentially fast with the dimension of $X$. When $X$ is a high dimensional process, high performance computing can help but it is well known that solving the least square problem does not scale well and then deteriorates the efficiency of the parallel implementation, see for instance~\cite{Pa11GPU,PaPi16}. \\

\JL{In this work, we target \emph{truly} path dependent options, i.e. options for which the payoff cannot be written as a function of a Markov process $X$ with reasonable size. In this case, \eqref{eq:EMarkov} does not hold anymore and computing the conditional expectation knowing $\cf_{t_k}$ becomes really challenging. The new idea proposed in this work consists in computing an approximation of $Z_{T_{k+1}}$ for which the conditional expectation knowing $\cf_{T_k}$ is known in a closed form. This will be achieved by using Wiener chaos expansion. Then, we rely on the orthogonality of the chaos expansion to introduce a high degree of parallelism in the algorithm.}

In Section~\ref{sec:wiener}, we briefly recall the general ideas sustaining Wiener chaos expansion and how it can be used to approximate conditional expectations. Then, we present our algorithm in Section~\ref{sec:ls-chaos} and explain how to efficiently implement it in parallel. Section~\ref{sec:cv} is devoted to the study of the convergence of the algorithm. We conclude with some numerical experiments in Section~\ref{sec:numerics}, which emphasize the impressive scalability of the the parallel implementation and the efficiency of the algorithm for some complex path dependent options.

\section*{Notation}

In this section, we gather some extensively used notation in the paper
\begin{itemize}
  \item For $\alpha \in \N^d$, $\abs{\alpha}_1 = \sum_{i=1}^d \alpha_i$. Similarly, for $\alpha \in (\N^n)^d$, $\abs{\alpha}_1 = \sum_{j=1}^d \sum_{i=1}^n \alpha_i^j$.
  \item For $\alpha \in \N^d$, $\alpha! = \prod_{i=1}^d \alpha_i!$. Similarly, for $\alpha \in (\N^n)^d$, $\alpha! = \prod_{j=1}^d \prod_{i=1}^n \alpha_i^j!$.
  \item For $d, n, p \in \N$, we define the set of multi-indices with total degree smaller than $p$ by \begin{align*}
    A^{\otimes d}_{p, n} & = \left\{ \alpha \in (\N^{n})^d\; : \abs{\alpha}_1 \le p
  \right\}
  \end{align*}
  \item For $d, n, p \in \N$, and $k \le n$  we define the set of multi-indices with total degree smaller than $p$ and no degree after $k$ by
  \begin{align*}
    \spaceIndex[k] = \left\{ \alpha \in \spaceIndex \; :
    \; \forall j \in \{1,\dots, d\}, \, \forall i > k, \; \alpha^j_i = 0 \right\}.
  \end{align*}
  \item For $i \in \N$, $H_i$ denote the $i-th$ Hermite polynomial.
  \item For $\alpha \in {(\N^n)}^d$, $x_1, \dots, x_n \in \R^d$, the multi-variate Hermite polynomials write
  \[
  H^{\otimes d}_\alpha (x_1, \dots, x_n) =
  \prod_{j=1}^d \prod_{i=1}^n H_{\alpha^j_i} (x^j_i).
  \]
\end{itemize}

\section{Wiener chaos expansion}
\label{sec:wiener}

\subsection{General framework}

In this section, we briefly recall the principles of Wiener chaos expansion and its basic properties. We refer to~\cite{nualart_98} for theoretical details. \\

Let $H_i$ be the $i-th$ Hermite polynomial defined by
\begin{align*}
  H_0(x) = 1; \qquad
  H_i(x) = (-1)^i \expp{x^2/2} \frac{d^i}{dx^i} (\expp{-x^2/2}),
  \mbox{ for } i \ge 1.
\end{align*}
They satisfy for all integer $i$, $H_i' = H_{i-1}$ with the convention $H_{-1} = 0$.
We recall that if $(X,Y)$ is a standard random normal vector in $\R^2$, $\E[H_i(X) H_j(Y)]  = i!\left( \E[XY] \right)^i \; \ind{i = j}$.

It is well-known that every square integrable $\cf_T$-measurable random variable $F$
admits the following orthonormal decomposition
\begin{align*}
  F = \E[F] + \sum_{\alpha \in (\N^\N)^d} \lambda_\alpha
  \prod_{j=1}^d \prod_{i \ge 1} H_{\alpha^j_i}\left(\int_0^T \eta^j_i(t) dB^j_t\right)
\end{align*}
where $\left((\eta^j_i)_{ 1 \le j \le d}\right)_{i \ge 1}$ is an orthonormal basis of $L^2([0,T], \R^d)$. We denote by $L_1^2([0,T], \R^d)$ the set of functions $f = (f_1, \dots, f_d) \in L^2([0,T], \R^d)$ such that for all $1 \le i \le d$, $\int_0^T f_i^2(t) dt = 1$.
For all $p \ge 0$, we define the Wiener chaos of order $p$ by
\begin{align*}
  \ch_p = \ClosedSpan \left\{ \prod_{j=1}^d H_{p_j}\left(\int_0^T f^j_t dB^j_t\right)\; : \; f
  \in L_1^2([0,T], \R^d), \; \sum_{j=1}^d p_j = p \right\}.
\end{align*}
We denote the projection of a random variable $F \in L^2(\cf_T)$ onto $\displaystyle
\bigoplus_{\ell=0}^p \ch_\ell$ by $C_p(F)$. Note that the spaces $\ch_\ell$ are orthogonal to each other thanks to the properties of the Hermite polynomials.

Consider the indicator functions of the grid defined by $0= t_0 < t_1 < \dots < t_n =
T$ with values in $\R^d$ defined by
\begin{align*}
  f_i^j(t) = \frac{\ind{]t_{i-1}, t_i]}(t)}{\sqrt{t_i - t_{i-1}}} \ee_j,
  \; i=1,\dots,n, \; j=1,\dots,d
\end{align*}
where $(\ee_1,\dots, \ee_d)$ denotes the canonical basis of $\R^d$.
Based on the definition of $\ch_p$, we introduce the \textit{truncated} Wiener chaos of order up to $p$
\begin{align*}
  \cc_{p,n} = \Span \left\{ H^{\otimes d}_{\alpha}(G_1, \dots, G_n) \; : \;
    \alpha \in (\N^n)^d, \, \abs{\alpha}_1 \le p \right\}
\end{align*}
where
\[
  H^{\otimes d}_\alpha (G_1, \dots, G_n) =
  \prod_{j=1}^d \prod_{i=1}^n H_{\alpha^j_i} (G^j_i) \quad \mbox{with}\quad  G^j_i = \frac{B^j_{t_i} - B^j_{t_{i-1}}}{\sqrt{t_i - t_{i-1}}}.
\]

From the orthogonality of the Hermite polynomials, we immediately deduce the following result.
\begin{proposition}
  \label{prop:chaos-trunc}
  Let $F$ be a real valued random variable in $L^2(\Omega, \cf_T, \P)$. Its $L^2$ projection onto $\cc_{p,n}$ writes
  \begin{align*}
    C_{p,n}(F) = \sum_{\alpha \in A^{\otimes d}_{p,n}} \lambda_\alpha
    H_{\alpha}^{\otimes d} (G_1, \dots, G_n)
  \end{align*}
  where
  \begin{align*}
    A^{\otimes d}_{p, n} & = \left\{ \alpha \in (\N^{n})^d\; : \abs{\alpha}_1 \le p
  \right\}
  \end{align*} and the coefficients $\lambda_\alpha$ are obtained as a dot product
  \begin{align}
    \label{eq:lambdaa}
    \lambda_\alpha = \inv{\alpha!} \, \E[F H_{\alpha}^{\otimes d} (G_1, \dots, G_n)].
  \end{align}
\end{proposition}

The random variable $C_{p,n}(F)$ is called the truncated chaos expansion of order $p$ of  the random variable $F$. With an obvious abuse of notation, we write, for  $\lambda \in \R^{A^{\otimes d}_{p, n}}$,
\begin{align}
  \label{eq:chaos-pd}
  C_{p,n}(\lambda) = \sum_{\alpha \in A^{\otimes d}_{p,n}} \lambda_\alpha
  H_{\alpha}^{\otimes d} (G_1, \dots, G_n).
\end{align}
We recall the main result concerning the convergence of the truncated chaos expansion (see Theorem~1.1.1 and Proposition~1.1.1 of~\cite{nualart_98})
\begin{proposition}
  \label{prop:chaos-cv}
  Let $F$ be a real valued random variable in $L^2(\Omega, \cf_T, \P)$. Then, $C_{p,n}(F)$ converges to $F$ in $L^2(\Omega, \cf_T, \P)$ when both $p$ and $n$ go to infinity.
\end{proposition}
The space of truncated Wiener chaos $\cc_{p,n}$ has the key property of being stable by the conditional expectation operator. More precisely, the following result explains how to compute, in a closed form, the conditional expectation of an element of $\cc_{p,n}$.
\begin{proposition}
  \label{prop:Et-chaos-d}
  Let $F$ be a real valued random variable in $L^2(\Omega, \cf_T, \P)$ and let $k
  \in \{1, \dots, n\}$ and $p \ge 0$
  \begin{align*}
    \E[C_{p,n}(F) | \cf_{t_k}] = \sum_{\alpha \in \spaceIndex[k]} \lambda_\alpha \;
    H^{\otimes d}_\alpha (G_1, \dots, G_{n})
  \end{align*}
  where $\spaceIndex[k]$ is the set of multi-indices vanishing after time $t_k$
  \begin{align*}
    \spaceIndex[k] = \left\{ \alpha \in \spaceIndex \; :
    \; \forall j \in \{1,\dots, d\}, \, \forall i > k, \; \alpha^j_i = 0 \right\}.
  \end{align*}
\end{proposition}

\begin{proof}
  Taking the conditional expectation in~\eqref{eq:chaos-pd} leads to
  \begin{align}
    \label{eq:Et-proof}
    \E[C_{p,n} (F) | \cf_{t_k}] = \sum_{\alpha \in A_{p,n}^{\otimes d}} \lambda_\alpha
    \left(\prod_{i=1}^k \prod_{j=1}^d H_{\alpha^j_i} (G^j_i)\right) \E\left[ \prod_{i=k+1}^n \prod_{j=1}^d H_{\alpha^j_i} (G^j_i) \Big|\cf_{t_k} \right].
  \end{align}
  Since the Brownian increments after time $t_k$ are independent of $\cf_{t_k}$ and are
  independent of one another, $\E\left[ \prod_{i=k+1}^n \prod_{j=1}^d H_{\alpha^j_i}
  (G^j_i) \Big|\cf_{t_k} \right] = \prod_{i=k+1}^n \prod_{j=1}^d \E\left[ H_{\alpha^j_i}
  (G^j_i) \right]$, which is zero as soon as $\sum_{i=k+1}^n \sum_{j=1}^d \alpha_i^j > 0$. Hence, the sum in~\eqref{eq:Et-proof} reduces to the sum over the set of multi-indices $\alpha \in A_{p,n}^{\otimes d}$ such that $\alpha^j_i = 0$ for all $i>k$ and $1 \le j \le d$, which is exactly the definition of the set $\spaceIndex[k]$.
\end{proof}

Since the sum appearing in $\E[C_{p,n}(F) | \cf_{t_k}]$ is reduced to a sum over the set of multi-indices $\alpha \in \spaceIndex[k]$, it actually only depends on the first $k$ increments $(G_1, \dots, G_k)$. One can easily check that $\E[C_{p,n}(F) | \cf_{t_k}]$ is actually given by the chaos expansion of $F$ on the first $k$ Brownian increments. Hence, computing a conditional expectation simply boils down to dropping the non measurable terms. While it may look like a naive way to proceed, it is indeed correct in this setting. To denote the chaos expansion on the time grid $(t_0, \dots, t_n)$ truncated to the first $k$ increments, we introduce the notation
\begin{align}
  \label{eq:Et-chaos-d}
   C_{p,n|k}(F) = \sum_{\alpha \in \spaceIndex[k]} \lambda_\alpha \;
    H^{\otimes d}_\alpha (G_1, \dots, G_{n}) = \E[C_{p,n}(F) | \cf_{t_k}].
\end{align}
With an obvious abuse of notation, we write for $\lambda \in \spaceIndex[k]$,
\[
  C_{p,n|k}(\lambda) =\sum_{\alpha \in \spaceIndex[k]} \lambda_\alpha \;
  H^{\otimes d}_\alpha (G_1, \dots, G_{n}).
\]

\subsection{Application to the approximation of conditional expectations}
\label{sec:wiener-cond}

In this section, we explain how to use the truncated Wiener chaos expansion of a random variable $F \in L^2(\Omega, \cf_T, \P)$, to compute its conditional expectation.

Assume that we need $M$ samples of the conditional expectations. We sample $M$ paths $(B^{(m)}_{t_1}, \dots, B^{(m)}_{t_n}, F^{(m)})$ of $(B_{t_1},\dots, B_{t_n}, F)$ and thanks to Proposition~\ref{prop:Et-chaos-d} we approximate $\E[F | \cf_{t_k}]$ on the sample path with index $m$ by
\begin{align*}
  C_{p,n|k}^{(m)}(\hl^M) = \sum_{\alpha \in \spaceIndex[k]} \hl_\alpha^M \;
    H^{\otimes d}_\alpha (G_1^{(m)}, \dots, G_{k}^{(m)})
\end{align*}
where
\begin{align*}
  \hl_\alpha^M = \inv{M \alpha! } \sum_{\ell = 1}^M F^{(\ell)} H^{\otimes d}_\alpha (G_1^{(\ell)}, \dots, G_{k}^{(\ell)}).
\end{align*}
Using the strong law of large numbers, we clearly have that for every $\alpha \in \spaceIndex[k]$, $\hl^M_\alpha$ converges a.s. to $\lambda_\alpha$ when $M$ goes to infinity. Then, we deduce that for any fixed $m$, $C^{(m)}_{p,n|k}(\hl^M)$ converges almost surely to  $C_{p,n|k}^{(m)}(\lambda)$ when $M \to \infty$.

\begin{remark}
  \label{rem:independent-samples}
 Note that we use the same samples to compute the coefficients of the chaos expansion $\hl^M_\alpha$ and to approximate $C^{(m)}_{p,n|k}$. It could have been possible to use different set of samples for the two parts and would have even simplified the theoretical analysis of the algorithm but the price to pay in terms of computational time is prohibitive. Using independent sets of samples would require to simulate new samples of the whole path at each date $T_k$.
\end{remark}

\section{The algorithm}
\label{sec:ls-chaos}

\subsection{Description of the algorithm}

We aim at solving the dynamic programming equation~\eqref{eq:policy-iteration} to obtain $\tau_1$. Then, the time$-0$ price of the Bermudan option writes
\begin{equation*}
  U_0 = \max(Z_0, \E[Z_{\tau_1}]).
\end{equation*}

For all $n \ge N$, consider a time grid $0 < t_0 < t_1 < \dots < t_n = T$ of $[0,T]$, such that $\{T_1, \dots, T_N\} \subset \{t_1, \dots, t_n\}$. We assume that $\lim_{n \to \infty} \sup_{0 \le k \le n-1} \abs{t_{k+1} - t_k} = 0$. For $k \le N$, we define $\sigma_k \in \N$ such that
\begin{align*}
  t_{\sigma_k} = T_k.
\end{align*}
Even though, we do not make the dependency on $n$ explicit, it is clear that $\sigma_k$ is an increasing function of $n$. \\

Now, we introduce some successive approximations of~\eqref{eq:policy-iteration}. First, we replace the true conditional expectation $E[Z_{\tau_{k+1}} | \cf_{T_k}]$ by the conditional expectation of the truncated Wiener chaos expansion of $Z_{\tau_{k+1}}$
\begin{equation*}
  \begin{cases}
    \tau_N^{p,n} = T_N \\
    \tau_k^{p,n} = T_k \ind{Z_{T_k} \ge C_{p, n|\sigma_k}(\lambda_k)} + \tau_{k+1}^{p,n} \ind{Z_{T_k} < C_{p, n|\sigma_k}(\lambda_k)}, \; \text{for $1 \le k \le N-1$}
  \end{cases}
\end{equation*}
where the $\lambda_k$'s are the coefficients of the truncated expansion of $Z_{\tau_{k+1}^{p,n}}$
\begin{align*}
  \lambda_{k, \alpha} &= \inv{\alpha!} \; \E[Z_{\tau_{k+1}^{p,n}} H^{\otimes d}_\alpha (G_1, \dots, G_{\sigma_k})] \quad \mbox{for } \alpha \in \spaceIndex[\sigma_k]
\end{align*}
The standard approach is to sample a bunch of paths of the model $S^{(m)}_{T_0}, S^{(m)}_{T_1},\dots, S^{(m)}_{T_N}$ along with the corresponding payoff paths  $Z^{(m)}_{T_0}, Z^{(m)}_{T_1},\dots, Z^{(m)}_{T_N}$, for $m=1,\dots,M$.  We denote by $B^{(m)}$ the Brownian path used to sample $S^{(m)}_{T_0}, S^{(m)}_{T_1},\dots, S^{(m)}_{T_N}$. Note that $B$ is sampled on the finer grid $t_0, \dots, t_n$, which enables us to deal with model discretization issues. The vector $G^{(m)}_1, .\dots, G^{(m)}_{n}$ corresponds to the increments of the Brownian motion $B$ on the finer time grid. To compute the $\tau_k$'s on each path, one needs to compute the conditional expectations $\E[Z_{\tau_{k+1}} | \cf_{T_k}]$ for $k=1,\dots,N-1$. Then, we introduce the final approximation of the backward iteration policy, in which the truncated chaos expansion is computed using a Monte Carlo approximation
\begin{equation*}
  \begin{cases}
    \htau_N^{p,n, (m)} = T_N \\
    \htau_k^{p,n, (m)} = T_k \ind{Z_{T_k}^{(m)} \ge C_{p, n|\sigma_k}^{(m)}(\hl_k^M)} + \htau_{k+1}^{p,n,(m)} \ind{Z_{T_k}^{(m)} < C_{p, n|\sigma_k}^{(m)}(\hl_k^M)}, \; \text{for $1 \le k \le N-1$}
  \end{cases}
\end{equation*}
where the $\hl_k^M$ are computed as described in Section~\ref{sec:wiener-cond}. For $k=1,\dots,N-1$, the vector $\hl_k^M$ is an element of $\R^{A_{p,n}^{\otimes d, \sigma_k}}$ and for every $\alpha \in A_{p,n}^{\otimes d, \sigma_k}$,
\begin{align}
  \label{eq:lambdaM}
  \hl_{k,\alpha}^M = \inv{M \alpha! } \sum_{\ell = 1}^M Z_{\htau_{k+1}^{p,n,(\ell)}}^{(\ell)} H^{\otimes d}_\alpha (G^{(\ell)}).
\end{align}
Then, we finally approximate the time$-0$ price of the option by
\begin{align}
  \label{eq:price-mc}
  U_0^{p,n,M} = \max\left(Z_0, \inv{M} \sum_{m=1}^M Z^{(m)}_{\htau_1^{p,n, (m)} }\right).
\end{align}
The pseudo code of our approach corresponds to Algorithm~\ref{algo-ls-chaos}.
\begin{algorithm2e}[ht]
  Generate $(G^{(1)}, Z^{(1)}), \dots, (G^{(M)}, Z^{(M)})$ $M$ i.i.d. samples following the law of $(Z_{t_i}, G_{t_i})_{1 \le i \le N}$ \;
  $\htau_N^{p,n,(m)} \gets T$ for all $m=1,\dots,M$ \;

  \For{$k=N-1,\dots,1$}{\label{a:outer}
    \For{$\alpha \in \spaceIndex[\sigma_k]$}{\label{a:expansion}
      $$
      \hl_{k,\alpha}^M = \inv{M \alpha! } \sum_{\ell = 1}^M Z_{\htau_{k+1}^{p,n,(\ell)}}^{(\ell)} H^{\otimes d}_\alpha (G^{(\ell)})
      $$
    }
    \For{$m=1,\dots,M$}{
      \begin{equation*}
          \htau_k^{p,n, (m)} = T_k \ind{Z_{T_k}^{(m)} \ge C_{p, n|\sigma_k}^{(m)}(\hl_k^M)} + \htau_{k+1}^{p,n,(m)} \ind{Z_{T_k}^{(m)} < C_{p, n|\sigma_k}^{(m)}(\hl_k^M)}
      \end{equation*}
    }
  }
  \begin{align*}
    U_0^{p,n,M} = \max\left(Z_0, \inv{M} \sum_{m=1}^M Z^{(m)}_{\htau_1^{p,n, (m)} }\right)
  \end{align*}
  \caption{Dynamic programming principle using Wiener chaos expansion}
  \label{algo-ls-chaos}
\end{algorithm2e}

\begin{remark}\label{rem:in-the-money}
  From a practical point of view, we advise to consider in the money paths in the chaos expansion as it was already noticed in~\cite{LS01}. Hence, the set $\{Z_{T_k}^{(m)} \ge C_{p, n|\sigma_k}^{(m)}(\hl_k^M)\}$ is replaced by $\{Z_{T_k}^{(m)}  > 0\} \cup \{Z_{T_k}^{(m)} \ge C_{p, n|\sigma_k}^{(m)}(\hl_k^M)\}$ and the coefficients of the chaos expansion are given by
  \[
  \hl_{k,\alpha}^M = \inv{M \alpha! } \sum_{\ell = 1}^M Z_{\htau_{k+1}^{p,n,(\ell)}}^{(\ell)} \ind{Z_{T_k}^{(\ell)}  > 0} H^{\otimes d}_\alpha (G^{(\ell)}).
  \]
  This modification does not change the theoretical analysis of the algorithm but improves its numerical behavior.
\end{remark}

Our algorithm is designed as a black box taking as inputs simulations of the Brownian motion and the corresponding payoff process. From a practical point of view, you can design the implementation in such as way that pricing a new product simply amounts to implementing the discretization of the model and the computation of the payoff.

\subsection{Comments on the algorithm}

The obvious and generic way to deal with truly path-dependent options or non Markovian model using the standard Longstaff Schwartz algorithm would be to consider the whole path as a regressor. It is very much unlikely that one can easily build a set of basis functions which are orthogonal for the law of the discretized path process. Hence, the regression problem would grow exponentially fast and as explained in~\cite{benguigui2012}, parallelism would not help much. Going beyond the Markovian setting requires an orthogonality property, which turns the regression problem into a series of independent inner-products. Of course, it is always possible to pretend everything is Markovian, but then you have no guaranty on the error you are making. \\

Our algorithm may be related to a \emph{regress later} method as investigated by~\cite{GlaYu04,BP18}. At time $T_k$, a regress later approach is typically composed of two steps: first $Z_{\tau_{k+1}}$ is decomposed on a set of $\cf_{T_{k+1}}$ measurable basis functions, which looks like a least squares approximation of the conditional expectation with respect to $\cf_{T_{k+1}}$.   Then, the conditional expectation of each basis function is computed analytically to obtain an approximation of $\E[Z_{\tau_{k+1}} | \cf_{T_k}]$. Our algorithm can also be seen as a two stage method: first we compute the chaos expansion of $Z_{\tau_{k+1}}$ and then we compute its conditional expectation. Although this way of formulating the algorithm is mathematically correct, it would be totally inefficient to implement it this way. As a matter of fact, taking the conditional expectation of a Wiener chaos expansion simply amounts to dropping non measurable terms and because every coefficient is computed on its own using an inner product, we can directly compute the conditional expectation of the chaos expansion by actually computing a chaos expansion with respect to the Brownian increments up to time $T_k$ only. This more pragmatic way of understanding our algorithm makes it actually closer to a \emph{regress now} approach.

\JL{Closely looking at Algorithm~\ref{algo-ls-chaos}, it is clear that the central part of the algorithm is the computation of the chaos expansion. Conveniently implementing this step plays a major role in the efficiency of the algorithm. In our C++ implementation, the chaos expansion is performed using the generic multivariate polynomial toolbox from~\cite{pnl}}.

\subsection{Complexity analysis}

Most of the computational time is spent computing the coefficients of the chaos expansions. Remember that the cardinality of $\spaceIndex[k]$ is given by $\binom{\sigma_k d + p}{n \sigma_k} = \frac{(\sigma_k d + p) \cdots (\sigma_k d + 1)}{p!}$. As the optimal policy is only updated on the in-the-money paths at each time step (see Remark~\ref{rem:in-the-money}), the complexity of iteration $k$ of the loop on line~\ref{a:outer} of Algorithm~\ref{algo-ls-chaos} is proportional to
\[
  \sharp\{\text{in-the-money paths at time } T_k\} \times \binom{\sigma_k d + p}{n \sigma_k}.
\]
It is worth noting that the complexity decreases when time decreases.
The order $p$ of the expansion plays a major role in the computational time of the algorithm. So, when the order of the expansion increases from $p$ to $p+1$, the computational time is multiplied by $\frac{\sigma_k d + p + 1}{p + 1}$.

\subsection{The parallel implementation}

\JL{The key computational trick of our algorithm is that the chaos coefficients $\lambda$ are written as independent expectations and can therefore be parallelized both across $\alpha$ and the number $M$ of Monte Carlo samples. Simply put, Algorithm~\ref{algo-ls-chaos} can be reduced to computing several independent Monte Carlo averages and is therefore very well suited for parallel programming.}  For a fixed time $T_k$, there are two ways of introducing parallelism.
\begin{enumerate}[(i)]
  \item The coefficients of the truncated Wiener chaos expansion can be computed in parallel. For two multi-indices $\alpha, \beta \in \spaceIndex[\sigma_k]$, the computations of $\hl^M_{k,\alpha}$ and $\hl^M_{k, \beta}$ are independent and can therefore be carried out simultaneously. The update of all the $\htau_k^{p,n,(m)}$ can also be performed in parallel. This approach looks very promising provided that the cardinality of $\spaceIndex[\sigma_k]$ is large enough, at least larger than the number of available computing resources. Note that
  \[
    \# \spaceIndex[\sigma_k] = \binom{\sigma_k \, d \;+\; p}{\sigma_k \, d}
  \]
  where we recall that $\sigma_k \to 0$ when $k \to 0$. This approach will be efficient for large enough $k$ but will inevitably fail to scale when $k$ decreases, ie for smaller dates.

  \item \label{para} Alternatively, we can use the number of Monte Carlo samples as the leverage for parallelism. Since the number of samples remains fixed during the whole algorithm, the parallelism will be as efficient for large $k$ as for small ones. Assume we have $R$ computing resources at our disposal, then each resource handles $M_R = M/R$ sample paths and runs the sequential algorithm~\ref{algo-ls-chaos} on these paths except that at each time step, a reduction followed by a broadcast are done right before updating the $\htau_k^{p,n,(m)}$, $m=1,\dots,M$. In this way, the chaos expansions are computed using the $M$ paths. We precisely describe this parallel algorithm in Algorithm~\ref{algo-ls-chaos-parallel}.
\end{enumerate}
We have followed the approach~\eqref{para} for our parallel implementation to make sure all the resources are always fully busy, which is the least requirement to ensure a decent scalability. The comparison of Algorithms~\ref{algo-ls-chaos} and~\ref{algo-ls-chaos-parallel} shows that the sequential and parallel algorithms differ very little. We even managed to merge the sequential and parallel implementations into a single code, which is hardly ever feasible especially when using MPI. Each computing resource samples a bunch of paths, on which it updates the optimal stopping policy and contributes to the computation of the $\hl^M_k$'s. \JL{At each time step, we compute an average (a \emph{reduction}) to get the value of the $\hl^M_k$'s and then we send (\emph{broadcast}) the coefficients to every resources. In practice, we actually use the \emph{AllReduce\footnote{See \url{https://mpitutorial.com/tutorials/mpi-reduce-and-allreduce/} for an explanation of how reduce and broadcast can be efficiently coupled.}} method from MPI.}

It was noted in~\cite{benguigui2012,PaPi16}, that using mini-batches to introduce parallelism in least-square Monte Carlo was not convincingly efficient, mainly because the backward induction is essentially sequential. This is due to the regression step itself, which cannot be solved efficiently when Monte Carlo paths are allocated by blocks on each processor. In any parallel implementation of least-squares Monte Carlo, the regression step eventually becomes the bottleneck because of its bad scalability. To circumvent this main issue with least-square Monte Carlo, \cite{Pa11GPU} replaced the regression step by a quantization approach, which allows for natural parallelism. Similarly, \cite{gobet2016stratified} used stratification to introduce conditional independence between the samples used in each strata. It may look as a mini-batch approach while they have to use new stratified sampling at each time step because they work in a backward stochastic differential setting. Hence, interpreting their approach in terms of mini-batches is not straightforward. In our approach, we rely on the orthogonality of the chaos expansion to replace the regression step by inner products computed by Monte Carlo. Therefore, our approach naturally fits into the mini-batch paradigm with no extra cost.

\SetKwBlock{Parallel}{In parallel do}{end}
\SetKw{Broadcast}{Broadcast}
\SetKw{Reduce}{Reduce}
\begin{algorithm2e}[htp!]
  $M_R \gets M / R$\;
  \Parallel{
    Generate $(G^{(1)}, Z^{(1)}), \dots, (G^{(M_R)}, Z^{(M_R)})$ $M_R$ i.i.d. samples following the law of $(Z_{t_i}, G_{t_i})_{1 \le i \le N}$ \;
    $\htau_N^{p,n,(m)} \gets T$ for all $m=1,\dots,M_R$ \;

    \For{$k=N-1,\dots,1$}{
      \For{$\alpha \in \spaceIndex[\sigma_k]$}{
        $$
        \hl_{k,\alpha}^{M_R} = \inv{M_R \alpha! } \sum_{\ell = 1}^{M_R} Z_{\htau_{k+1}^{p,n,(m)}}^{(\ell)} H^{\otimes d}_\alpha (G^{(\ell)})
        $$
      }
      \Reduce the $\hl_{k,\alpha}^{M_R}$ to obtain $\hl_{k,\alpha}^{M}$ \;
      \Broadcast $\hl_{k,\alpha}^{M}$ for $\alpha \in \spaceIndex[\sigma_k]$ \;
      \For{$m=1,\dots,M_R$}{
        \begin{equation*}
            \htau_k^{p,n, (m)} = T_k \ind{Z_{T_k}^{(m)} \ge C_{p, n|\sigma_k}^{(m)}(\hl_k^{M})} + \htau_{k+1}^{p,n,(m)} \ind{Z_{T_k}^{(m)} < C_{p, n|\sigma_k}^{(m)}(\hl_k^M)}
        \end{equation*}
      }
    }
  \begin{align*}
    U_1^{p,n,M_R} = \inv{M_R} \sum_{m=1}^{M_R} Z^{(m)}_{\htau_1^{p,n, (m)} }
  \end{align*}
  }
  \Reduce the $U_1^{p,n,M_R}$ \;

  $U_0^{p,n,M} = \max\left(Z_0, U_1^{p,n}\right)$ \;

  \caption{Parallel algorithm for solving the dynamic programming principle using Wiener chaos expansion}
  \label{algo-ls-chaos-parallel}
\end{algorithm2e}

\section{Convergence of the algorithm}
\label{sec:cv}

In this section, we basically follow the lines of the methodology introduced in~\cite{clp02}. The statements of the convergence results are quite similar even if some assumptions had to be modified to match our framework, but the proofs differ to adapt to the new formulation of the regression step. 

There are two independent parts in this section. In Section~\ref{sec:cv-chaos}, we study the convergence of the algorithm with respect to the chaos expansion when all expectations are assumed to be computed exactly (no Monte Carlo approximation). In Section~\ref{sec:cv-mc}, we fix the order and the discretization used in the chaos expansion and we study the convergence with respect to the number of Monte Carlo samples. This is achieved by first proving that the Monte Carlo approximations of the chaos expansion at each time step converge to the true coefficients.

\subsection{Notation}

To avoid over expanding notation, we simply write $G$  instead of $(G_1, \dots, G_n)$ in the chaos expansions. At some points, it may be important to make precise which Brownian increments are used in the chaos expansion. To do so, we introduce the notation
\begin{align*}
  C_{p,n}(\lambda; G) = \sum_{\alpha \in A^{\otimes d}_{p,n}} \lambda_\alpha
  H_{\alpha}^{\otimes d} (G).
\end{align*}

 First, it is important to note that the paths $\tau_1^{p,n,(m)},\dots,\tau_N^{p,n,(m)}$ for $m=1,\dots,M$ are identically distributed but not independent since the Monte Carlo computation of the chaos expansion coefficients $\hl_k^M$ mixes all the paths. We define the vector $\Lambda$ of the coefficients of the successive expansions $\Lambda = (\lambda_1, \dots, \lambda_{N-1})$ and its Monte Carlo approximation $\hL^M = (\hl_1^M,.\dots,\hl_{N-1}^M)$.

Now, we recall the notation used by~\cite{clp02} to study the convergence of the original Longstaff Schwartz approach.\\
Given a parameter $\ell = (\ell_1, \dots, \ell_{N-1})$ in $\R^{\spaceIndex[\sigma_1]} \times \dots \times \R^{\spaceIndex[\sigma_{N-1}]}$ and vectors $z = z_1,\dots,z_N$ in $\R^N$ and $g = (g_1, \dots, g_n)$ in $(\R^d)^n$, we define the vector field $F= F_1, \dots, F_N$ by
\begin{align*}
  \begin{cases}
    F_N(\ell, z, g) &= z_N \\
    F_k(\ell, z, g) &= z_k \ind{z_k \ge C_{p,n|\sigma_k}(\ell; g)} + F_{k+1}(\ell, z, g) \ind{z_k < C_{p,n|\sigma_k}(\ell; g)}, \; \text{for $1 \le k \le N-1$}.
  \end{cases}
\end{align*}
Note that $F_k(\ell, z, x)$ does not depend on the first $k-1$ components of $\ell$, ie $\ell_1, \dots, \ell_{k-1}$. Moreover,
\begin{alignat*}{2}
  &F_k(\Lambda, Z, G) &&= Z_{\tau_k^{p,n}}, \\
  &F_k(\hL^M, Z^{(m)}, G^{(m)}) &&= Z^{(m)}_{\htau_k^{p,n, (m)}}.
\end{alignat*}
For $k=1, \dots, N$, we also define the functions $\phi_k : \R^{\spaceIndex[\sigma_1]} \times \dots \times \R^{\spaceIndex[\sigma_{N-1}]} \to \R$ and $\psi_k : \R^{\spaceIndex[\sigma_1]} \times \dots \times \R^{\spaceIndex[\sigma_{N-1}]} \to \R^{\spaceIndex[\sigma_{k}]}$ by
\begin{align*}
  \phi_k(\ell) = \E[F_k(\ell, Z, G)]  & \quad \text{and} \quad
  \psi_k(\ell) = \left(\E[F_k(\ell, Z, G) H_{\alpha}^{\otimes d}(G)]\right)_{\alpha \in \spaceIndex[\sigma_{k}]}.
\end{align*}
Note that $\phi_k$ and $\psi_k$ actually only depends on $\ell_k,\dots,\ell_{N-1}$ but not on the first $k-1$ components of $\ell$.

\subsection{Chaos approximation of conditional expectations}
\label{sec:cv-chaos}

\begin{proposition}
	\label{prop:cv-prix-pn}
  For all $k=1,\dots,N$, $\lim_{p,n \to \infty} \E[Z_{\tau^{p,n}_k} | \cf_{T_k}] = \E[Z_{\tau_k} | \cf_{T_k}]$ in $L^2(\Omega)$.
\end{proposition}
\begin{proof}
  We proceed by induction. The result is true for $k=N$ as $\tau_N = \tau^{p,n}_k = T$. Assume it holds for $k+1$ ($k \le N-1$), we will prove it is true for $k$.
  \begin{align*}
    & \E[Z_{\tau^{p,n}_k} - Z_{\tau_k} | \cf_{T_k}] \\
    & = Z_{T_k} \left(\ind{Z_{T_k} \ge C_{p,n|\sigma_k}(\lambda_k)} - \ind{Z_{T_k} \ge \E[Z_{\tau_{k+1}}|\cf_{T_k}]}\right) \\
    & \quad + \E\left[Z_{\tau_{k+1}^{p,n}} \ind{Z_{T_k} < C_{p,n|\sigma_k}(\lambda_k)}   - Z_{\tau_{k+1}} \ind{Z_{T_k} < \E[Z_{\tau_{k+1}}|\cf_{T_k}]}| \cf_{T_k}\right] \\
    & = (Z_{T_k} - \E[Z_{\tau_{k+1}}|\cf_{T_k}]) \left(\ind{Z_{T_k} \ge C_{p,n|\sigma_k}(\lambda_k)} - \ind{Z_{T_k} \ge \E[Z_{\tau_{k+1}}|\cf_{T_k}]}\right) \\
    & \quad + \E\left[Z_{\tau_{k+1}^{p,n}} - Z_{\tau_{k+1}} | \cf_{T_k}\right] \ind{Z_{T_k} < C_{p,n|\sigma_k}(\lambda_k)}.
  \end{align*}
  By the induction assumption, the term $\E\left[Z_{\tau_{k+1}^{p,n}} - Z_{\tau_{k+1}} | \cf_{T_k}\right]$ goes to zero in $L^2(\Omega)$ as $p,n$ both go to infinity. So, we just have to prove that
  \[
  A_k = (Z_{T_k} - \E[Z_{\tau_{k+1}}|\cf_{T_k}]) \left(\ind{Z_{T_k} \ge C_{p,n|\sigma_k}(\lambda_k)} - \ind{Z_{T_k} \ge \E[Z_{\tau_{k+1}}|\cf_{T_k}]}\right)
  \]
  converges to zero in $L^2(\Omega)$.
  \begin{align}
    \label{eq:Ak}
    \abs{A_k} &\le \abs{Z_{T_k} - \E[Z_{\tau_{k+1}}|\cf_{T_k}]} \abs{\ind{Z_{T_k} \ge C_{p,n|\sigma_k}(\lambda_k)} - \ind{Z_{T_k} \ge \E[Z_{\tau_{k+1}}|\cf_{T_k}]}}  \notag\\
    & \le \abs{Z_{T_k} - \E[Z_{\tau_{k+1}}|\cf_{T_k}]} \abs{ \ind{\E[Z_{\tau_{k+1}}|\cf_{T_k}] >Z_{T_k} \ge C_{p,n|\sigma_k}(\lambda_k)} - \ind{C_{p,n|\sigma_k}(\lambda_k) > Z_{T_k} \ge \E[Z_{\tau_{k+1}}|\cf_{T_k}]}} \notag\\
    & \le \abs{Z_{T_k} - \E[Z_{\tau_{k+1}}|\cf_{T_k}]} \ind{\abs{Z_{T_k} - \E[Z_{\tau_{k+1}}|\cf_{T_k}]} \le  \abs{C_{p,n|\sigma_k}(\lambda_k) - \E[Z_{\tau_{k+1}}|\cf_{T_k}]}} \notag\\
    & \le \abs{C_{p,n|\sigma_k}(\lambda_k) - \E[Z_{\tau_{k+1}}|\cf_{T_k}]} \notag\\
    & \le \abs{C_{p,n|\sigma_k}(\lambda_k) - C_{p,n|\sigma_k}(\E[Z_{\tau_{k+1}}|\cf_{T_k}])}  +
    \abs{C_{p,n|\sigma_k}(\E[Z_{\tau_{k+1}}|\cf_{T_k}]) - \E[Z_{\tau_{k+1}}|\cf_{T_k}]}.
  \end{align}
  Note that $C_{p,n|\sigma_k}(\lambda_k)  = C_{p,n|\sigma_k}(\E[Z_{\tau_{k+1}}^{p,n}|\cf_{T_k}])$. The truncated chaos expansion $C_{p,n|\sigma_k}$ being an orthogonal projection on the space of random variables measurable with respect to the Brownian increments $G_1, \dots, G_k$, we clearly have that
  \begin{align*}
    &\E\left[\abs{C_{p,n|\sigma_k}(\lambda_k) - C_{p,n|\sigma_k}(\E[Z_{\tau_{k+1}}|\cf_{T_k}])}^2\right] \\
    & \le \E\left[\abs{\E[Z_{\tau_{k+1}^{p,n}}|\cf_{T_k}] - \E[Z_{\tau_{k+1}}|\cf_{T_k}]}^2\right] \\
    & \le \E\left[\abs{\E[Z_{\tau_{k+1}^{p,n}}|\cf_{T_{k+1}}] - \E[Z_{\tau_{k+1}}|\cf_{T_{k+1}}]}^2\right]
  \end{align*}
  where the last inequality comes from the orthogonal projection feature of the conditional expectation. Then, the induction assumption for $k+1$ yields that $C_{p,n|\sigma_k}(\lambda_k) - C_{p,n|\sigma_k}(\E[Z_{\tau_{k+1}}|\cf_{T_k}])$ goes to zero in $L^2(\Omega)$ as $p,n$ go to infinity. So, the first term on the r.h.s of~\eqref{eq:Ak} goes to zero.

  As $C_{p,n|\sigma_k}(\E[Z_{\tau_{k+1}}|\cf_{T_k}]) = C_{p,n}(\E[Z_{\tau_{k+1}}|\cf_{T_k}])$, the second term on the r.h.s of~\eqref{eq:Ak} goes to zero in $L^2(\Omega)$ thanks to Proposition~\ref{prop:chaos-cv}. Combining these two results yields the convergence statement of the proposition.
\end{proof}

\begin{remark}
  When the discrete time payoff process $(Z_{T_k})_{0 \le k \le N}$ is measurable for the filtration generated by the discrete time Brownian increments $(\cg_k)_{0 \le k \le N} = (\sigma(B_{T_{i+1}} - B_{T_i}, i \le k))_{0 \le k \le N}$, the result of Proposition~\ref{prop:cv-prix-pn} simplifies to $\lim_{p \to \infty} \E[Z_{\tau^{p,N}_k} | \cf_{T_k}] = \E[Z_{\tau_k} | \cf_{T_k}]$ in $L^2$. There is no need to let $n$ go to infinity, it is sufficient to take $n = N$. From a practical point of view, one should choose $n$ in order to monitor the discretization error between the true payoff process $Z$ and its implementable discretization $Z^n$ on a time grid with $n$ steps. Then, the parameter $n$ has to be considered as being fixed and we actually compute the price of the Bermudan option with payoff process $Z^n$ instead of $Z$. Therefore, when the model can be exactly sampled, one should choose $n = N$.
\end{remark}

\subsection{Convergence of the Monte Carlo approximation}
\label{sec:cv-mc}

In the following, we assume that $p$ and $n$ are fixed and we study the convergence with respect to the number of samples $M$. 

\subsubsection{Strong law of large numbers}
\label{sec-slln}

To start, we prove the convergence of the coefficients of the chaos expansions.
\begin{proposition}
  \label{prop:lambdaM}
Assume that for every $k=1,\dots,N$, $\P(Z_{T_k} \in \cc_{p,n|\sigma_k}) = 0$. Then, for every $k=1,\dots,N$, $\hL^M_k$ converges to $\Lambda_k$ a.s. as $M \to \infty$.
\end{proposition}
The proof of Proposition~\ref{prop:lambdaM} based on the following key lemma from~\cite{clp02}. The assumption $\P(Z_{T_k} \in \cc_{p,n}) = 0$ may look surprising but a very similar assumption was already required in~\cite[Lemma~3.2]{clp02}. This assumption combined with the following lemma proves that the vector field $F(a, Z, G)$ is a.s. continuous w.r.t the expansion coefficients~$a$.

\begin{lemma}
  \label{lem:Flip}
  For every $k=1,\dots, N-1$,
  \[
    \abs{F_k(a, Z, G) - F_k(b, Z, G)} \le \left(\sum_{i=k}^N \abs{Z_{T_i}}\right) \left(\sum_{i=k}^{N-1} \ind{\abs{Z_{T_i} - C_{p,n|\sigma_i}(b_i)} \le \abs{a_i - b_i} \norm{C_{p,n|\sigma_i}}}\right)
  \]
  where
\begin{align*}
  \norm{C_{p,n}} = \sup_{\abs{\lambda} = 1} \abs{C_{p,n}(\lambda)}.
\end{align*}
\end{lemma}

\begin{proof}[Proof of Proposition~\ref{prop:lambdaM}]
  We proceed by induction.  For $k=N-1$, the result directly follows from the standard strong law of large numbers. Choose $k \le N-2$ and assume the result holds for $k+1,\dots,N-1$ , we aim at proving this is true for $k$.
  \begin{align*}
   \hl_{k,\alpha}^M = \inv{M \alpha!} \sum_{m = 1}^M F_{k+1}(\hL_{k+1}^M, Z^{(m)}, G^{(m)}) H^{\otimes d}_\alpha (G^{(m)}).
  \end{align*}
  By the standard strong law of large number, $\inv{M \alpha!} \sum_{m = 1}^M F_{k+1}(\hL_{k+1}, Z^{(m)}, G^{(m)}) H^{\otimes d}_\alpha (G^{(m)})$ converges a.s. to $\inv{\alpha!}\E[ F_{k+1}(\hL_{k+1}, Z, G) H^{\otimes d}_\alpha (G)] = \l_{k,\alpha}$. Then, it is sufficient to prove that
  \[
    \Psi_M = \inv{M} \sum_{m=1}^M \left(F_{k+1}(\hL^M_{k+1}, Z^{(m)}, G^{(m)})  - F_{k+1}(\hL_{k+1}, Z^{(m)}, G^{(m)})\right)  H^{\otimes d}_\alpha (G^{(m)})\to 0 \quad a.s.
  \]
  Then, using Lemma~\ref{lem:Flip}, we have
  \begin{align*}
    & \abs{\Psi_M}  \le \inv{M} \sum_{m=1}^M  \abs{F_{k+1}(\hL^M_{k+1}, Z^{(m)}, G^{(m)}) - F_{k+1}(\hL_{k+1}, Z^{(m)}, G^{(m)})}  \abs{H^{\otimes d}_\alpha (G^{(m)})} \\
    & \le \inv{M} \sum_{m=1}^M  \sum_{i=k+1}^N \abs{Z_{T_{i+1}}^{(m)}} \left(\sum_{i=k+1}^{N-1} \ind{\abs{Z_{T_i}^{(m)} - C_{p,n|\sigma_i}^{(m)}(\Lambda_i)} \le \abs{\hL^M_i  - \Lambda_i}  \norm{C_{p,n|\sigma_i}}}\right) \abs{H^{\otimes d}_\alpha (G^{(m)})}
   \end{align*}
   From the induction assumption for $k+1, \dots, N-1$, we have that for $i=k+1,\dots, N-1$, $\hL^M_i \to \Lambda_i$. Then, for any $\varepsilon >0$, we have
   \begin{align*}
     & \limsup_M \abs{\Psi_M}  \\
     & \le \limsup_M \inv{M} \sum_{m=1}^M  \sum_{i=k+1}^N \abs{Z_{T_{i+1}}^{(m)}} \left(\sum_{i=k+1}^{N-1} \ind{\abs{Z_{T_i}^{(m)} - C_{p,n|\sigma_i}^{(m)}(\Lambda_i)} \le \varepsilon  \norm{C_{p,n|\sigma_i}}}\right) \abs{H^{\otimes d}_\alpha (G^{(m)})} \\
     & \le E\left[\sum_{i=k+1}^N \abs{Z_{T_{i+1}}} \left(\sum_{i=k+1}^{N-1} \ind{\abs{Z_{T_i} - C_{p,n|\sigma_i}(\Lambda_i)} \le \varepsilon  \norm{C_{p,n|\sigma_i}}}\right) \abs{H^{\otimes d}_\alpha (G)}\right]
   \end{align*}
   where the last equality follows from the strong law of large numbers. As $\P(Z_{T_k} \in \cc_{p,n|\sigma_k}) = 0$ for all $k$, we can let $\varepsilon$ go to $0$ to obtain that $\limsup_M \abs{\Psi_M}  = 0$ a.s.
\end{proof}
Once the convergence of the expansion is established, we can now study the convergence of $U^{p,n,M}_0$ to $U^{p,n}_0$ when $M \to \infty$.
\begin{theorem}
  \label{thm:slln}
  Assume that for every $k=1,\dots,N$, $\P(Z_{T_k} \in \cc_{p,n}) = 0$. Then, for $q=1,2$ and all $k=1,\dots,N$,
  \[
  \lim_{M \to \infty} \inv{M} \sum_{m=1}^M \left(Z_{\htau_k^{p,n,(m)}}^{(m)}\right)^q = \E\left[\left(Z_{\tau_k^{p,n}}\right)^q\right] \quad a.s.
  \]
\end{theorem}
\begin{proof}
  Note that $\E[(Z_{\tau_k^{p,n}})^q] = \E[F_k(\hL, Z, G)^q]$ and by the strong law of large numbers
  \[
    \lim_{M \to \infty} \inv{M} \sum_{m=1}^M F_k(\hL, Z^{(m)}, G^{(m)})^q = \E[F_k(\hL, Z, G)^q] \quad a.s.
  \]
  Hence, we have to prove that
  \[
    \Delta F_M =  \inv{M} \sum_{m=1}^M \left(F_k(\hL^M, Z^{(m)}, G^{(m)})^q - F_k(\hL, Z^{(m)}, G^{(m)})^q\right) \xrightarrow[M \to \infty]{a.s} 0.
  \]
  For any $x, y \in \R$, and $q=1,2$, $\abs{x^q - y^q}= \abs{x - y} \abs{x^{q-1} + y^{q-1}}$.
  Using Lemma~\ref{lem:Flip} and that $\abs{F_k(\gamma, z, g)} \le \max_{k \le j \le N} \abs{z_j}$, we have
  \begin{align*}
    & \abs{\Delta F_M}  \le \inv{M} \sum_{m=1}^M  \abs{F_k(\hL^M_k, Z^{(m)}, G^{(m)})^q - F_k(\hL_k, Z^{(m)}, G^{(m)})^q} \\
    & \le 2  \inv{M} \sum_{m=1}^M  \sum_{i=k}^N \max_{k \le j \le N} \abs{Z_{T_j}^{(m)}}\abs{Z_{T_{i+1}}^{(m)}} \left(\sum_{i=k}^{N-1} \ind{\abs{Z_{T_i}^{(m)} - C_{p,n|\sigma_i}^{(m)}(\Lambda_i)} \le \abs{\hL^M_i  - \Lambda_i}  \norm{C_{p,n|\sigma_i}}}\right)
   \end{align*}
   Using Proposition~\ref{prop:lambdaM}, $\hL^M_i \to \Lambda_i$ for all $i=1,\dots,N-1$. Then for any $\varepsilon > 0$,
   \begin{align*}
     & \limsup_M \abs{\Delta F_M} \\
     & \le 2 \limsup_M \inv{M} \sum_{m=1}^M  \sum_{i=k}^N \max_{k \le j \le N} \abs{Z_{T_j}^{(m)}} \abs{Z_{T_{i+1}}^{(m)}} \left(\sum_{i=k}^{N-1} \ind{\abs{Z_{T_i}^{(m)} - C_{p,n|\sigma_i}^{(m)}(\Lambda_i)} \le \varepsilon  \norm{C_{p,n|\sigma_i}}}\right) \\
     & \le 2 \E\left[ \sum_{i=k}^N \max_{k \le j \le N} \abs{Z_{T_j}} \abs{Z_{T_{i+1}}} \left(\sum_{i=k}^{N-1} \ind{\abs{Z_{T_i} - C_{p,n|\sigma_i}(\Lambda_i)} \le \varepsilon  \norm{C_{p,n|\sigma_i}}}\right) \right]
   \end{align*}
   where the last inequality follows from the strong law of large numbers as $\E[\max_{k \le j \le N} \abs{Z_{T_j}}^2 ] < \infty$. We conclude that $\limsup_M \abs{\Delta F_M} = 0$ by letting $\varepsilon$ go to $0$ and by using that for every $k=1,\dots,N$, $\P(Z_{T_k} \in \cc_{p,n}) = 0$.
\end{proof}
The case $q=1$ proves the strong law of large numbers for the algorithm. Considering that all the paths are actually mixed through the chaos expansion, it is unlikely that the estimators $\inv{M} \sum_{m=1}^M Z_{\htau_k^{p,n,(m)}}^{(m)}$ for $k=1,\dots,N$ are unbiased. We recall that $U_k^{p,n,M} =  \inv{M} \sum_{m=1}^M F_k(\hL^M, Z^{(m)}, G^{(m)})$ and $Z_{\tau_k^{p,n}} = F_k(\Lambda, Z, G)$. Then,
\begin{align*}
  & \E\left[U_k^{p,n,M}\right]  - \E\left[Z_{\tau_k^{p,n}} \right] = \E\left[ \inv{M} \sum_{m=1}^M \left(F_k(\hL^M, Z^{(m)}, G^{(m)}) - F_k(\Lambda, Z^{(m)}, G^{(m)})\right)\right] \\
& = \E\left[ F_k(\hL^M, Z^{(1)}, G^{(1)}) - F_k(\Lambda, Z^{(1)}, G^{(1)}) \right]
\end{align*}
where we have used that all the random variables have the same distribution. Hence, the bias of our estimator is directly linked to the gap between $\hL^M$ and the true value $\L$. Let $p < p'$, then for any $\alpha \in \spaceIndex$, $\alpha \in {A^{\otimes d}_{p',n}}$ and the corresponding value $\hl^M_{k, \alpha}$ is the same for $p$ and $p'$. This means that when $p$ increases, the length of $\hL^M$ increases with the first components remaining unchanged. Therefore, $\abs{\hL^M - \L}$ increases with $p$, which suggests that, for a fixed $M$, the bias also increases with $p$. Moreover, it was already noted in \cite{glasserman04number} that for a fixed number of samples $M$, the mean square error on the coefficients of the regression explodes with the number of regressors. In our framework, this means that, for a fixed M, $\E\left[\abs{\hL^M - \L}^2\right]$ will increase with $p$.

\subsubsection{Discussion on the rate of convergence}

From Theorem~\ref{thm:slln}, we deduce that the standard empirical variance estimator applied to our algorithm converges. For every $k=1,\dots,N$,
\begin{equation}
  \label{eq:empirical-var}
  \lim_{M \to \infty} \quad \inv{M} \sum_{m=1}^M \left(Z_{\htau_k^{p,n,(m)}}^{(m)}\right)^2 - \left(\inv{M} \sum_{m=1}^M Z_{\htau_k^{p,n,(m)}}^{(m)} \right)^2 = \Var(Z_{\tau_k^{p,n}}) \quad a.s.
\end{equation}
The convergence rate analysis carried out in~\cite{clp02} applies steadily to our approach. Then, under suitable assumptions, the vector
\begin{equation}
  \left(\sqrt{M} \left(\inv{M} \sum_{m=1}^M Z_{\htau_k^{p,n,(m)}}^{(m)} - \E[Z_{\tau_k^{p,n}}] \right)\right)_{k=1,\dots, N}
\end{equation}
converges in law to a normal distribution with mean zero. As noted in~\cite{clp02}, determining the asymptotic variance directly from the data generated by a single run of the algorithm is almost impossible. From the proof of the central limit theorem for their algorithm, we have, when $M$ goes to infinity, in the $L^2$ sense
\begin{multline}
  \label{eq:var-decomposition}
  \sqrt{M} \left(\inv{M} \sum_{m=1}^M Z_{\htau_k^{p,n,(m)}}^{(m)} - \E[Z_{\tau_k^{p,n}}] \right)\\ = \sqrt{M} \left(\inv{M} \sum_{m=1}^M Z_{\tau_k^{p,n,(m)}}^{(m)} - \phi_k(\Lambda) \right)  + \sqrt{M} ( \phi_k(\hL^M) - \phi_k(\Lambda)).
\end{multline}
Remember that $Z_{\tau_k^{p,n,(m)}}^{(m)}  = F_k(\Lambda, Z^{(m)}, G^{(m)})$. By the standard central limit theorem, $\sqrt{M} \left(\inv{M} \sum_{m=1}^M Z_{\tau_k^{p,n,(m)}}^{(m)} - \phi_k(\Lambda) \right)$ converges in law to a normal distribution with variance $\Var(Z_{\tau_k^{p,n}})$. Then, using the empirical variance of the estimator as a measurement of the algorithm converge actually misses part of the variance since from~\eqref{eq:empirical-var}, we know that the empirical variance only takes into account the first term on the r.h.s of~\eqref{eq:var-decomposition}.

\section{Numerical experiments}
\label{sec:numerics}

In this section, we carry out several numerical experiments using our algorithm. In the different tables, the ``Price'' column corresponds to the value of $U_0^{p,n,M}$ averaged over $25$ independent runs of the algorithm and the ``Variance'' column is the variance of $U_0^{p,n,M}$ computed on these $25$ independent runs. The first two experiments, which deal with put options, enable us to compare the accuracy of our method with the standard Longstaff Schwartz algorithm using only the in-the-money paths at each time step, whose price is reported in the ``LS'' column. Then, we consider more sophisticated truly path dependent options for which the use of the standard Longstaff Schwartz algorithm becomes prohibitive because of the well-known curse of dimensionality. In all the examples, we use $N = n$, ie we do not subdiscretize the grid given by the exercising dates to compute the chaos expansions.

\subsection{Examples in the Black Scholes model}

The $d-$dimensional Black Scholes model writes for $j \in \{1, \dots, d\}$
\begin{align*}
  dS^j_t = S^j_t ( r_t dt + \sigma^j L_j dB_t)
\end{align*}
where $B$ is a Brownian motion with values in $\R^d$, $\sigma = (\sigma^1, \dots,
\sigma^d)$ is the vector of volatilities, assumed to be deterministic and positive at
all times and $L_j$ is the $j$-th row of the matrix $L$ defined as a square root of the
correlation matrix $\Gamma$, given by
\begin{equation*}
  \Gamma = \begin{pmatrix}
    1 & \rho & \hdots & \rho\\
    \rho & 1 &\ddots & \vdots\\
    \vdots&\ddots&\ddots& \rho\\
    \rho &\hdots & \rho & 1
  \end{pmatrix}
\end{equation*}
where $\rho \in ]-1 / (d-1), 1]$ to ensure that $\Gamma$ is positive definite.

\subsubsection{Assessing the method on the one-dimensional put option}

Before investigating more elaborate numerical example, we want to test our method on the Bermudan put option. As standard as this example might be, getting a trustworthy reference price is not an easy task. We rely on prices computed by a convolution method in~\cite{Oosterlee08} and later used as reference prices in~\cite{oosterlee09}. We report in Table~\ref{tab:put-bs} our values compared to the reference prices for two different volatilities. Our prices are already very close the \emph{true} prices even with a second order expansion $p=2$. On these examples, we are within $0.2\%$ of the reference prices.

\begin{table}[htp]
  \centering\begin{tabular}{cccccccc}
    \hline
    $\sigma$ & p & M & Price & Variance & Reference price\\
    \hline
    0.2  & 2 & 1E5 & 10.48  & 7E-4  & 10.4795  \\
    0.2  & 2 & 1E6 & 10.47  & 7E-5  \\
    0.2  & 3 & 1E5 & 10.48  & 6E-4  \\
    0.2  & 3 & 1E6 & 10.47 & 6E-5 \\
    0.25 & 2 & 1E5 & 11.96  & 1E-3 & 11.987 \\
    0.25 & 2 & 1E6 & 11.94  & 2E-4  \\
    0.25 & 3 & 1E5 & 11.96  & 9E-4  \\
    0.25 & 3 & 1E6 & 11.96  & 1E-4  \\
    \hline
  \end{tabular}
	\caption{Put option with $r=0.1$, $T=1$, $K=110$, $S_0=100$ and $N=10$.} \label{tab:put-bs}
\end{table}

\subsubsection{A put basket option}

We consider a put basket option with payoff
\[
  \left(K - \sum_{i=1}^d \omega_i S_T^i\right)_+,
\]
which can be priced using the classical Longstaff Schwartz algorithm and therefore enables us to test the accuracy of our approach in a multidimensional setting. We test our algorithm in dimension $5$ and report the results in Table~\ref{tab:basket} for different numbers of samples $M$ and different orders $p$ of chaos expansion. The values reported in the ``LS'' column correspond to the prices computed with the Longstaff Schwartz algorithm with $10^6$ samples and using as regression functions the set of polynomials of total order $3$ completed with the payoff function.

We notice that an expansion of order $p=2$ already gives a price fairly close to the ``LS'' one for a quite reasonable computational time. Increasing $p$ to $3$ improves the accuracy only when the number of samples $M$ is also increased. \JL{Indeed, we can see that the prices obtained for $K=90$ and $p=3$ for small values of $M$ ($M=5E4$ or $M=1E5$) are above the dual price. This clearly happens because $p$ is too large compared to $M$, which induces a bias.} We refer the reader to the discussion following Theorem~\ref{thm:slln} for more information on this point. Hence, in a brand new setting, we advise to start with $p=2$ and to monitor the variance to fix how many Monte Carlo samples are required $M$. Then, if need be, one can try $p=3$ with keeping in mind that $M$ should be increased at the same time. In our example, we basically add an order of magnitude to $M$, when going from $p=2$ to $p=3$.
\begin{table}[htp]
  \centering\begin{tabular}{cccccccccc}
    \hline
    T & K & N & p & M & Price & Variance & LS & Dual price\\
    \hline
    3 & 100 & 20 & 2 & 5E4 &  4.01793 &  0.00039 & 4.07 & 4.3\\
    3 & 100 & 20 & 2 & 1E5 &  4.00769 &  0.00028 \\
    3 & 100 & 20 & 2 & 1E6 &  3.99801 &  2.15E-05 \\
    3 & 100 & 20 & 3 & 5E4 &  4.2544 &  0.00041 \\
    3 & 100 & 20 & 3 & 1E5 &  4.1965 &  0.00024 \\
    3 & 100 & 20 & 3 & 1E6 &  4.06587 &  2.19E-05 \\
    3 & 90 & 20 & 2 & 5E4 &  1.29423 &  0.00013 & 1.32 & 1.47\\
    3 & 90 & 20 & 2 & 1E5 &  1.27274 &  0.00011 \\
    3 & 90 & 20 & 2 & 1E6 &  1.25166 &  2.242E-05 \\
    3 & 90 & 20 & 3 & 5E4 &  1.52426 &  8.84E-05 \\
    3 & 90 & 20 & 3 & 1E5 &  1.49847 &  0.00010 \\
    3 & 90 & 20 & 3 & 1E6 &  1.31845 &  2.72E-05 \\
    \hline
  \end{tabular}
	\caption{Basket option with $r=0.05$, $d=5$, $\sigma^i = 0.2$, $\omega^i = 1/d$, $S_0^i=100$ and $\rho=0.2$.} \label{tab:basket}
\end{table}

\subsubsection{Asian option}

For this example, we consider a one dimensional Black Scholes model, $d=1$.
We consider an Asian with payoff $Z_t = (K - X_t)_+$ with $X_0 = S_0$ and for $t >0$
\[
  X_t = \inv{t} \int_{0}^{t} S_u du.
\]

We approximate the continuous time integral by an arithmetic average and we compare our results with the one reported by~\cite{HW93} (in the ``HW'' column in Table~\ref{tab:asian}), which, despite being quite old, is still considered as a benchmark by many papers investigating American Asian options.

\begin{table}[htp]
  \centering\begin{tabular}{cccccccc}
    \hline
    T & K & N & p & M & Price & Variance & HW\\
    \hline
    1 & 45 & 20 & 2 & 1E6 & 8.55   & 1E-4 & 8.55 \\
    1 & 45 & 20 & 3 & 1E6 & 8.47   & 1E-4 \\
    1 & 45 & 20 & 3 & 1E7 & 8.61   & 3E-6 \\
    1 & 50 & 20 & 2 & 1E6 & 4.81   & 1E-4 & 4.89\\
    1 & 50 & 20 & 3 & 1E6 & 4.7    & 1E-4 \\
    1 & 50 & 20 & 3 & 1E7 & 4.79   & 4E-6 \\
    2 & 45 & 20 & 2 & 1E6 & 10.63  & 2E-4 & 10.62\\
    2 & 45 & 20 & 3 & 1E6 & 10.46  & 2E-4 \\
    2 & 45 & 20 & 3 & 1E7 & 10.66  & 6E-6 \\
    2 & 50 & 20 & 2 & 1E6 & 7.28   & 2E-4 & 7.33\\
    2 & 50 & 20 & 3 & 1E6 & 7.24   & 2E-4 \\
    2 & 50 & 20 & 3 & 1E7 & 7.29   & 7E-6 \\
    \hline
  \end{tabular}
	\caption{Asian option with $r=0.1$, $d=1$, $\sigma = 0.3$, $S_0=50$ and $N=40$ (resp. $80$) for $T=1$ (resp. $T=2$).} \label{tab:asian}
\end{table}

It is known that although the payoff does not seem to be Markovian in dimension $1$, if we augment the state space and consider the pair $(S, X)$, then the option becomes Markovian again. Hence, Asian options can serve as a good example to assess the efficiency of our algorithm by considering the non Markovian representation of the Asian option in our method. As in the previous example, we notice that a second order expansion $p=2$ already gives very accurate price, within $1\%$ of the benchmark price computed by~\cite{HW93} using a tree method. Increasing $p$ to $3$ does not significantly improve the accuracy of the process but does require to increase the number of Monte Carlo samples. 

\subsubsection{Moving average option}

For this example, we consider a one dimensional Black Scholes model, $d=1$.
We consider a moving average option with payoff $Z_t = (S_t - X_t)_+$ for $t\ge \delta + \ell$ with
\[
  X_t = \inv{\delta} \int_{t-\delta - \ell}^{t-\ell} S_u du
\]
where $\delta >0$ is the length of the averaging window and $\ell$ is a delay.

We approximate the continuous time integral by an arithmetic average and compare our results with the benchmark prices computed by~\cite{BTW11}. Let $N_\delta = \frac{\delta}{T} N$ and $N_\ell = \frac{\ell}{T} N$. For every $T_i \ge \delta + \ell$, we approximate $X_{t_i}$ by
\[
	X^N_{T_i} = \inv{N_\delta} \sum_{j=i - N_\delta - N_\ell + 1}^{i - N_\ell} S_{T_j}.
\]
The benchmark prices reported in the ``LS'' column come from~\cite{BTW11} and were computed using the standard Longstaff Schwartz algorithm with regression factors at time $T_i$ given by
\[
  \left(S_{T_{i - N_\delta - N_\ell + 1}}, S_{T_{i - N_\delta - N_\ell + 2}}, \dots, S_{T_{i - N_\ell}}\right).
\]
This leads to a regression problem with $N_\delta$ variables, which makes it very CPU demanding. While our approach may also look like a multi variate regression, the main difference lies in the choice of an orthogonal basis function which turns the computation of the coefficients of the regression from a linear system into a bunch of independent Monte Carlo computations. Although this seems a minor change, it is indeed a huge improvement as it breaks the bottleneck of the standard Longstaff Schwartz algorithm and makes it easy to parallelize.

We run two series of tests on the moving average option, which is a typical example of a true path-dependent option in the sense that the size of the underlying Markov process $X$ (see~\eqref{eq:EMarkov}) is basically the number of exercising dates. We report in Table~\ref{tab:moving-average-nodelay} the results for the non delayed option, ie $N_\ell = 0$ and in Table~\ref{tab:moving-average-delay} the results for the option with delay. When there is no delay (Table~\ref{tab:moving-average-nodelay}), we are able to recover the prices computed with the Longstaff Schwartz method using the full list of regressors.  Our results are already very accurate for a chaos expansion of order $p=2$. To really benefit from a more accurate chaos expansion of order $p=3$, one also needs to increase the number of samples $M$ to cut down the bias. Note the price $>4.268$ in the ``LS-price'' column for $w=0.04$. In~\cite{BTW11}, they did not succeed in computing the price of this option using the Longstaff Schwartz method using the full list of regressors, so they only provided a non Markovian approximation $4.268$, which is always below the true price. Hence, the value $4.30329$ obtained for $p=3$ and $M=10^6$ does make sense. We also report in the column ``Dual price'' of Table~\ref{tab:moving-average-nodelay} the upper bound obtained from~\cite{lel:chaos}. A small gap remains between the lower and upper bounds, but it can be considered as more than acceptable considering the numerical challenges represented by the highly path-dependent products. \JL{Since the work by~\cite{BTW11}, new methods have been developed to handle high dimensional regressions mostly by using machine learning techniques. For instance, we can cite the recent works~\cite{becker2019deep,becker2019deep2}. It would be very interesting to use these algorithms to build a price comparator for the non-Markovian settings.}

\begin{table}[htp]
	\centering\begin{tabular}{ccccccc}
	\hline
	$\delta$  &  p & M       & Price   & Variance   & LS-Price & Dual price\\
	\hline
	0.02  & 2 & 1E5  & 3.53118 & 8.97E-06 & 3.531 & 3.76 \\
	0.02  & 2 & 1E6 & 3.53863 & 9.7E-07 \\
	0.02  & 3 & 1E5  & 3.45177 & 7.05E-06 \\
	0.02  & 3 & 1E6 & 3.52758 & 7.12E-07 \\
	0.04 & 2 & 1E5  & 4.30318 & 1.7E-04 & $>$ 4.268 & 4.52 \\
	0.04 & 2 & 1E6 & 4.31781 & 8.82E-07  \\
	0.04 & 3 & 1E5  & 4.18467 & 1.31E-04 \\
	0.04 & 3 & 1E6 & 4.30239 & 1.10E-06 \\
	\hline
	\end{tabular}
	\caption{Moving average option with $S_0 = 100$, $\sigma = 0.3$, $r=0.05$, $T=0.2$, $N=n=50$ and $\ell=0$ (no delay).} \label{tab:moving-average-nodelay}
\end{table}

\begin{table}[htp]
	\centering\begin{tabular}{ccccc}
		\hline
		  p & M & Price   & Variance    \\
		\hline
		 2 & 5E4   & 6.62011 & 7.5E-4 \\
		 2 & 1E5  & 6.67733 & 2.5E-4 \\
		 2 & 1E6 & 6.74565 & 2.00E-05 \\
		 3 & 5E4   & 6.28484 & 4.2E-4 \\
		 3 & 1E5  & 6.36383 & 3.1E-4 \\
		 3 & 1E6 & 6.65446 & 8.02E-06 \\
		\hline
	\end{tabular}
	\caption{Moving average option with $S_0 = 100$, $\sigma = 0.3$, $r=0.05$, $T=0.2$, $N=n=50$, $\ell=0.08$ ($N_\ell=20$) and $\delta = 0.02$ ($N_\delta = 5)$.} \label{tab:moving-average-delay}
\end{table}

\subsection{A put option in the Heston model}

We start with a put option in the Heston model to assess the accuracy of our algorithm. We recall the definition of the Heston model
\begin{align*}
  dS_t &= S_t(r_t dt + \sqrt{\sigma_t} (\rho dW^1_t + \sqrt{1 - \rho^2} dW^2_t)) \\
  d\sigma_t &=  \kappa (\theta - \sigma_t) dt + \xi \sqrt{\sigma_t} dW^1_t.
\end{align*}

\begin{table}[htp]
  \centering\begin{tabular}{ccccc}
    \hline
    d & p & M & Price & Variance\\
    \hline
    1 & 2 & 1E5 &  1.71756 &  4.68E-05 \\
    1 & 2 & 1E6 &  1.69802 &  7.68E-06 \\
    1 & 2 & 1E7 &  1.69699 &  4.37E-07 \\
    1 & 3 & 1E5 &  1.73389 &  8.43E-05 \\
    1 & 3 & 1E6 &  1.72354 &  6.63E-06 \\
    1 & 3 & 1E7 &  1.72274 &  8.53E-07 \\
    \hline
  \end{tabular}
	\caption{Put option in the Heston model with $S_0 = K = 100$, $T=1$, $\sigma_0 = 0.01$, $\xi=0.2$, $\theta=0.01$, $\kappa = 2$, $\rho = -0.3$, $r=0.1$, $N=20$} \label{tab:put-heston}
\end{table}
For the put option used in the numerical experiments of Table~\ref{tab:put-heston}, the Longstaff Schwartz algorithm gives $1.74$ using degree $3$ polynomials for the regression and $10^6$ samples. Note that as we only consider in the money paths for the regression step, the payoff function is actually a linear function of the underlying asset --- a degree one polynomial. So there is no need to add the payoff function to the regression basis as for more sophisticated options. Obviously, we consider both the asset price and the volatility process as regression factors.

\JL{Clearly, we see in the figures of Table~\ref{tab:put-heston} that going from $1E6$ to $1E7$ samples does not make any difference on the prices. On the contrary, the prices obtained with $M=1E5$ are always a little higher, which may look surprising. This is a actually related to the bias phenomenon described at the end of Section~\ref{sec-slln}. The variance of the coefficients of the chaos expansion is responsible for introducing a bias into the price. To avoid this, one needs to use sufficiently many Monte Carlo samples to compute the chaos expansions. Anyway, the prices obtained with $p=3$ and $M=1E6$ or $M=1E7$ are within $1\%$ of the standard Longstaff Schwartz price.}

\subsection{Scalability of the parallel implementation}

 The scalability tests were run on a BullX DLC supercomputer containing 3204 cores. The code is written in C++ using the OpenMPI library to handle the communication and the PNL library~\cite{pnl} to compute the chaos expansions in a generic way for any order $p$. We report in Table~\ref{tab:scalability} the evolution of the efficiency with respect to the number of resources used. We recall that the efficiency is defined as the ratio between the sequential running time and the product of the parallel running time times the number of resources. Clearly, the efficiency takes values between $0$ and $1$ and the closer to one, the better. In the example used for the scalability study, we managed to cut down the computational time from an hour and a half to 14 seconds while maintaining the efficiency at almost $0.7$, which represents an astonishing improvement in terms of scalability. For a fixed size problem, it is well-known that the efficiency eventually decreases to zero when the number of processors go to infinity as every algorithm has a purely sequential part which becomes predominant in the end. Hence, the efficiency value of $0.7$ has to be considered together with the absolute computational time. We refer to~\cite{BaudeBossy08,dbggs_10} for experiments on the scalability of different parallel approaches for Bermudan options. Although their framework is a bit different from ours, we can assert that our $0.7$ efficiency proves a very good scalability.
\begin{table}[htp]
  \centering
  \begin{tabular}{ccc}
    \hline
    \#Procs & Time (sec.) & Efficiency \\
    \hline
    1  &  4768  &  1 \\
    2  &  2402  &  0.99 \\
    4  &  1234  &  0.97 \\
    16  &  353  &  0.84 \\
    32  &  173  &  0.86 \\
    64  &  89  &  0.84 \\
    128  &  47  &  0.79 \\
    256  &  24  &  0.76 \\
    512  &  14  &  0.68 \\
    \hline
  \end{tabular}
  \caption{Scalability of the parallel algorithm on the moving average option with delay used of Table~\ref{tab:moving-average-delay} with $M=10^6$ and $p=3$.} \label{tab:scalability}
\end{table}

\section{Conclusion}

In this work, we have presented a new algorithm to price Bermudan option in non Markovian settings: the non Markovian feature can either come from the truly path dependent feature of the option or from the use of rough volatility models for instance. Our algorithm makes it easy to design a generic American option pricer, actually not more difficult than for a European option pricer. Although this may sound a bit ambitious, our algorithm is designed as a black box taking as inputs sample paths of the underlying multi-dimensional Brownian motion and the associated samples of the payoff process, which is basically the same as for European options. The smart design of our algorithm combined with orthogonality feature of the Wiener chaos expansion leads to an embarrassingly parallel algorithm, in which each node samples a bunch of paths, on which it updates the optimal stopping policy. Each node contributes to the computation of the $\hl^M_k$'s and at each time step, we make a reduction to get the value of the $\hl^M_k$'s and then a broadcast makes the coefficients available to everyone. The parallel implementation requires very few communications and therefore shows an impressive efficiency.

The methodology developed in this work in a Brownian setting could be adapted to Lévy processes by adding Charlier polynomials to the Hermite polynomials. We refer to~\cite{LabGei17} for the use of chaos expansions with jumps.

\bibliographystyle{abbrvnat}
\bibliography{biblio.bib}
\end{document}